\documentclass[journal]{IEEEtran}
\makeatletter
\def\endthebibliography{%
  \def\@noitemerr{\@latex@warning{Empty `thebibliography' environment}}%
  \endlist
}
\makeatother

\usepackage{cite}
\ifCLASSINFOpdf
\else
\fi
\usepackage{color}
\usepackage{subfig}
\usepackage{float}
\usepackage{graphicx}
\usepackage{amsmath}
\usepackage{amsthm}
\usepackage{mathptmx}
\usepackage{amssymb}

\newtheorem{thm}{Theorem}

\theoremstyle{definition}

\usepackage{algorithm,algpseudocode}
\usepackage{url}
\begin{document}
\title{Vulnerability  Assessment of Power Grids Based on Both Topological and Electrical Properties}

\author{Cunlai~Pu~and~Pang~Wu
        % <-this % stops a space

\thanks{C. Pu and P. Wu are with the School of Computer Science and Engineering, Nanjing University of Science and Technology, Nanjing 210094, China. Email: pucunlai@njust.edu.cn. }% <-this % stops a space
}

\maketitle

\begin{abstract}
In modern power grids, a local failure or attack can trigger  catastrophic cascading failures, which make it challenging to assess the attack vulnerability of power grids.  In this Brief, we define the $K$-link attack problem  and study the attack vulnerability  of power grids under cascading failures. Particularly, we propose a link centrality measure based on both topological and electrical properties of power grids. According to this centrality, we propose a greedy attack algorithm and an optimal attack algorithm.  Simulation results on  standard IEEE bus test data show  that the optimal attack  is better than the  greedy attack   and the traditional PSO-based  attack  in fracturing  power grids. Moreover, the greedy attack  has  smaller computational complexity  than the optimal attack and  the PSO-based  attack with an adequate attack efficiency. Our work helps to understand the vulnerability of power grids and provides some clues for securing power grids.

\end{abstract}

\begin{IEEEkeywords}
Vulnerability assessment, power grid, link centrality, cascading failure, network attacks.
%IEEE, IEEEtran, journal, \LaTeX, paper, template.
\end{IEEEkeywords}

% \fi
\IEEEpeerreviewmaketitle

\section{Introduction}
%\IEEEPARstart{T}{his} demo file is intended to serve as a ``starter file''
%for IEEE journal papers produced under \LaTeX\ using
%IEEEtran.cls version 1.8b and later.

%\IEEEPARstart{A}{s} the scale and complexity of mordern power systems grows unprecently, their reliability becomes a serious issure and attacts more and more attention.
    Nowadays, power grids in the real world face various kinds of risks such as natural disasters and attacks. Even  worse, a local failure in power grids can  result in large-scale blackouts \cite{schafer2018dynamically,yang2017small}. A recent example is the nationwide recurring electrical blackouts in Venezuela began in March 2019, which was supposed to be caused  by a local vegetation fire and cyberattacks. The catastrophic cascades of failures  pose a great threat to human life and national security.  Thus, it is of great importance to understand and control  the cascading failures of  power grids.

 In the past decade, the complex network theory has been widely applied to the study of  cascading failures \cite{barabasi2016network}. On the modelling side,  power grids can be abstracted as  interdependent networks, and then the percolation theory has been used to explore the dynamics of cascading failures \cite{gao2012networks}. Rich behaviors have been observed when taking the power grid as interdependent network.  For instance, Buldyrev et al. \cite{buldyrev2010catastrophic} found a hybrid phase transition (HPT), where the order parameter has both a jump and a critical scaling.

  On the controlling side, researchers proposed many optimization strategies against  cascading failures in power grids.  Tu et al. \cite{tu2018optimal} used the simulated annealing method to optimize the network topology, and found that it is better to make the network sparsely connected, and place the generators as decentralized hubs. They further investigated the  weak interdependency between networks of cyber-physical systems (CPS) and discussed how the failure propagation probabilities affect the robustness of CPS \cite{tu2019robustness}.
  %Zhang et al. use the memetic algorithm (MA) to optimize the coupling links of interconnected networks.
   Chen et al. \cite{chen2017robustness} performed the critical node analysis  to identify the vital nodes  in terms  of network robustness.  They found that assortative coupling of node destructiveness is more robust in densely coupled networks, whereas disassortative coupling of node robustness and node destructiveness is better in sparsely coupled networks.
Zhong et al. \cite{zhong2019restoration} studied the repair process against cascading failures  by considering the optimization  of repair resource, timing and load tolerance, for different coupling strength and network topologies of interdependent networks.
Zhu et al. \cite{zhu2018optimization} established two multiobjective optimization models that consider both the operational cost of
links and the robustness of  networks.
Zhang et al. \cite{zhang2018optimal}  employed the particle swarm optimization (PSO) algorithm to optimize the defense resource allocation to improve the network robustness.

To better understand and control cascading failures, we need to explore the role  of individual nodes and links in power grids. When quantifying the importance of nodes or links, the complex network theory only focuses on  network topology information \cite{barabasi2016network,jiang2016effects,jiang2018cascade}. However,  the electrical features of nodes and links are profound \cite{azzolin2018electrical}. Particularly,  a link of small topological importance might have large current load. The broken of this kind of link  has significant impact on the function of power grids.   It is thus more reasonable  to consider both  topological and electrical features of power grids when characterizing the importance of nodes and links.

%The impact that a power transmission system topology and associated electrical features have on overall system reliability is still not well understood.
In this brief, we study vulnerability assessment of power grids under cascading failures. We define the importance of links based on both topological and electrical features, and remove a few important links as the initial attack that triggers cascading failures.   Our main contributions are as follows:
\begin{itemize}
\item We propose a link centrality measure, which combines  link degree and link current. The weights of the two features are tunable and represented by two variables. This centrality is better than link degree  and link current in quantifying the importance of links in power grids.

\item According to the link centrality, we propose a greedy attack algorithm and an optimal attack algorithm. These attack algorithms are designed to cause large-scale cascading failures, based on which we can assess the vulnerability of power grids by simulation.

%{\color{red}{\item We study how the performance of the proposed attack algorithms is affected by the network heterogeneity. We find the optimal  parameter of the power-law degree distribution, which leads to the maximum attack efficiency.(problem)}}
\end{itemize}

In the next section, we  present the cascading failure model and related metrics. In section \uppercase\expandafter{\romannumeral3}, we introduce our link centrality measure and the  parameter tuning method.  In section \uppercase\expandafter{\romannumeral4}, we define the link attack problem and provide our  attack algorithms. In section \uppercase\expandafter{\romannumeral5}, we show the simulation results and related analysis.  Section \uppercase\expandafter{\romannumeral6} is our conclusion.

\section{Model and  Measurement}
In a  power grid, power stations and transmission lines can be abstracted as nodes and links, respectively. Then, we obtain the network topology of power grid, denoted by $G= (V, E) $,  where $V$ is the node set with $N$ nodes, and $E$ is the link set consisting of $M$ links.
 %Based on the Kirchhoff's law, we can calculate  can calculate the actual operation status of each node and link in power grid.When the network is stable, the scale of cascade failures can be measured by the number of failed nodes.

\subsection{ Current model}
Usually,  there are generator node, consumer node, distribution node, and transformer
node in a power grid. Here, following Ref. \cite{tu2018optimal} and for the purpose of simplification, we only consider two kinds of nodes: generator  node $i$ and consumer node $j$.
 %and  As described in model [13], the grid model considers four kinds of nodes: 1) node $l$ is the power generation node as a fixed voltage source, 2) node $p$ is the power consumption consumer node, 3) node $k$ is the power distribution node, 4) node $h$ is the transformer node.
 Then,  the Kirchhoff's current law equation for a power grid is written as
\begin{equation}
\label{eqn_example}
Y*{[\cdots\ v_i\ v_j\  \cdots]} = {[\cdots\ v_i\ I_j\  \cdots]},
\end{equation}
in which
 \begin{equation}
 Y=\left[\begin{matrix} \ldots&\ldots&\ldots&\ldots\\\ldots&y_i&0&\ldots\\\ldots&-Y_{ji}&Y_{jj}&\ldots\\ \ldots&\ldots&\ldots&\ldots&   \end{matrix} \right],
 \end{equation}
and $v_i$ and $v_j$ represent the voltages of  generator node and consumer node, respectively.  $I_j$ is the current of consumer node $j$. In matrix  $Y$, $y_i=1$. $Y_{ij}$ is the admittance of link $ (i, j)$, and $Y_{ij} =0$ if nodes $i$ and $j$ are not connected. Also, we have $Y_{jj}\!\!=-\!\!\sum_{i \neq j}Y_{ji}$.
When the admittances of transmission lines, the voltages
of generator nodes and the current consumptions of
consumer nodes are given, the voltage of each node can be
computed by Eq. (1). Then, the current flowing through
link $(i, j)$ can be calculated as $ I_{ij} = (v_i-v_j)*Y_{ij}$.

\subsection{Cascading failure model}
 Assume the load of node $i$ is $L{(i)} = u_i*I_{oi}$, where $I_{oi}$ is the total current flowing out of  node $i$, and the load of link $(i, j)$ is the current flowing through it, i.e., $I_{ij}$. The maximum load  a node can bear is set to be $1+\alpha$ times of its original load, and the maximum load of link $(i, j)$ is assumed to be $1+\beta$ times of its original current,  where parameters $\alpha$ and $\beta$ are the safety margins of nodes and links, respectively.  Note that the original state of nodes or links corresponds to the case when the power grid operates normally, that is there is no attacks or failures. In the cascading failure model \cite{buldyrev2010catastrophic}, there is usually an initial attack, e.g.  randomly removing a  node or link. This initial event will cause the load change of the other nodes and links, especially for those close to the area of initial attack. When the load of a node or link exceeds its maximum allowed  value, it will break, which further causes the load change of the other nodes and links.
The detailed steps in the cascading failure simulation process are as follows \cite{tu2018optimal}:

%\romannumeral1) The node admittance matrix of the initial network is determined according to $Y$. According to formula (1), the initial voltage of the node and the current of the link can be calculated.

\romannumeral1) Calculating the initial load and maximum load of each component (node and link) in the power  grid.

\romannumeral2) Randomly removing a component.

\romannumeral3) The grid topology changes due to the  removal, recalculating the load of each component. A component is set to be broken if its load exceeds its maximum.

\romannumeral4) After removing the failed components, the network  splits into several small subgraphs. If a subgraph does not contain a  generator node, all nodes in this subgraph are set to be invalid nodes.

\romannumeral5)  Repeating steps \romannumeral3) -\romannumeral4)~  until the load of all  remaining components is no greater than the maximum.

\subsection{ Vulnerability  measurement}
In power systems, the outage scale is usually measured by the number of failed nodes. Following Ref. \cite{tu2018optimal},   the damage that is caused by component set $i$  is quantified as
\vspace{-4px}
\begin{equation}
\label{eqn_example}
\Phi(i) = \frac{N_{unserved}(i)}{N},
\end{equation}
where $N_{unserved}(i)$ is the number of total failed nodes after the cascading failure caused by the initial removal of component set $i$.  Note that the failed nodes contain  the overloaded ones and  those in the subnetwork  of no  generator nodes. Obviously, the larger the damage, the more critical the component set is to the power grid.

\section{Link centrality measure and its solution }
Nodes and links are the main components in power grids. Node centrality has been widely discussed in the literature. Here we study the link centrality which is relatively  less discussed, but critical to the vulnerability assessment of power grids. Previously, many topology based  link centralities were developed \cite{barabasi2016network}.
However,  the joint effect of topological and electrical properties of links  on the vulnerability of power girds is still not well understood.
%Only based on topological indicators or electrical indicators can not effectively find the key links and key link sets in the power grid. Therefore, this paper proposes a comprehensive index method, which comprehensively considers electrical and topological indicators.It is proved that the problem of finding the comprehensive index parameter set belongs to the $NP$-$hard$ problem. In order to effectively find the better attribute matrix parameter set, the solution method of the best index based on particle swarm is given.

\subsection{Our link centrality measure}
We quantify the centrality of links based on both topological and electrical features. Specifically, we consider the link degree and link current. The degree of a link is the number of links (except itself) incident to its two end nodes. The current of a link is the rate of flow of electric charge pasting it. Note that during the cascading failure, the degree and current of a link might change due to the broken of overloaded components. Since we focus on the initial attacks,  we use the original state of power grid to quantify the centrality of links. Then, the centrality of link $(i, j)$ is defined as
\begin{equation}
\Theta_{ij}=h_1D_{ij}+h_2I_{ij},
\end{equation}
where $D_{ij}$ and $I_{ij}$ are the initial degree and current of link $(i, j)$, respectively. For  link current, we ignore  its direction and use its absolute value.  $h_1$ and $h_2$,  ranging in $(-\infty,  \infty)$,  are the weights of link degree and link current. In real applications, we usually need to find the optimal values of $h_1$ and $h_2$, which is a NP-hard problem.

\subsection{ Parameter tuning  based on PSO }
We use  the particle swarm optimization algorithm (PSO) \cite{kennedy2010particle}, to search the optimal parameters of our link centrality measure.  Compared with other heuristic algorithms, PSO has powerful global search ability and is easy to implement.
In this algorithm, there are $m$$(>0)$ particles moving on the network. Particle $i$ is a candidate solution of parameters, $X_i=[h_1, h_2]$, and  corresponds to a local optimal value $p^i_{best}$. For all the particles,  there is a global optimal value $g_{best}=max\{p^i_{best}\}$.
 The particle position is continuously updated according to the following equations:
\vspace{-4px}
\begin{equation}
\label{eqn_example}
{v_i}^{k+1}=w{v_i}^k+c_1r_1(p^i_{best}-{x_i}^k)+c_2r_2(g_{best}-{x_i}^k)
\end{equation}
\vspace{-4px}
\begin{equation}
\label{eqn_example}
{x_i}^{k+1}={x_i}^k+{v_i}^k
\end{equation}
\vspace{-4px}
\begin{equation}
\label{eqn_example}
w =w_0-iter/iter_{max}
\end{equation}
In Eq. (5), $r_1$ and $r_2$  are random numbers in [0,1]. $c_1$ is the cognitive coefficient and $c_2$ is the social learning coefficient. In Eq. (6),  ${v_i}^k$ and ${x_i}^k$ are the velocity and position of the $i$th particle in the $k$th iteration. In Eq. (7), $w$ is the inertia coefficient. Large inertia coefficient helps  particles  jump out of the local optimum, while  small inertia coefficient is beneficial to the local accurate search and  convergence of the algorithm. $w_0$ is the initial  inertia coefficient.  $iter$  and $ iter_{max}$ respectively represent  the current number of iterations and the maximum allowed number of iterations.

 %For the PSO algorithm, it is easy to prematurely converge and the latter algorithm is easy to generate an oscillation phenomenon near the optimal value. The linearly varying weight coefficient is used to linearly reduce the inertia coefficient from the maximum value to the minimum value.

\section{Application of link centrality in vulnerability assessment }
 We apply our link centrality to the vulnerability assessment of power grids. In the assessment, we usually simulate  network attacks and the consequent cascading failures, and then calculate the  damage caused by the attacks.  Different kinds of attacks result in different damage. Here, we consider the link attacks. The simplest link attack strategy is random attack, in which we randomly remove a certain fraction of links. More efficient link attack strategies are desired in the vulnerability assessment.

\subsection{ Problem Definition}
Given an integer $K$,  the problem is to find a set of $K$ links ($1 \leq K \leq M$), the removal of  which will cause the maximum damage to the power grid (abbreviated as $KLS$ problem). Note that the damage is measured as the percentage of failed nodes after the cascading failure (see Eq. (3)), which is triggered by the initial removal of the $K$-link set.

\begin{thm}
The $KLS$ problem is $NP$-$Complete$.
\end{thm}
\begin{proof}
 Given a set of $K$ links,  we can calculate the percentage of failed nodes after the cascading failure triggered by the removal of the set of links in polynomial time, which means that the $KLS$ problem is  $NP$. Moreover, the $KLS$ problem can be reduced to the 0/1 knapsack problem \cite{freville2004multidimensional}. In this problem, given a set of items, each with a weight and a value, we determine the number of each item, 0 or 1,  to include in a collection so that the total weight is less than or equal to a given limit and the total value is as large as possible.  In the $KLS$ problem, each link can cause some damage, and can be selected or not. The number of selected links is fixed to be $K$. The task is to determine the $K$-link set that achieves the largest damage. Since  0/1 knapsack problem is $NP$-$Complete$, so is the $KLS$ problem.
\end{proof}

\subsection{Optimal attack based on PSO}
Since the ${KLS}$ problem can be reduced to the 0/1 knapsack problem and is $NP$-$Complete$, we employ the PSO to search the optimal $K$-link set. In the  $m$ particles, particle $i$ is set to be $X_i=[x_1, x_2, \ldots, x_M]$, where $x_j$ corresponds to the $j$th link of the power grid. If this link is selected to be removal, $x_j=1$; otherwise, $x_j=0$. The constraint is $\sum_{j=1}^{M} x_j=K$. The particles update their velocity and position iteratively based on Eqs. (5)-(7) until finding the approximately optimal solution of the   ${KLS}$ problem.
 Since the duration of single cascade failure is unable to estimate, we only consider the number of cascading failures  when discussing the time complexity.  The   PSO based optimal attack (PSO-OA) has the time complexity of $O(m*iter_{max})$.  The pseudocode of this algorithm is in Algorithm 1.
\begin{algorithm}[htb]
  \caption{PSO based optimal attack (PSO-OA)}
  \label{alg:Framwork}
  \begin{algorithmic}
    \Require
     Adjacency matrix $G$,  power generator node set $Q$, and  integer $K$
    \Ensure
     The total number of selected links is fixed to be $K$.
 \\ \textit{Initialisation}: Randomly generate $m$ particles; each particle contains $M$ elements; randomly set $K$ elements to be 1 and the rest to 0.
     \For{$i=1, $$iter_{max}$}
     \For{$j=1, m$}
      \State Calculate ${\Phi(X_j)}$ for each particle
      \EndFor
       \State Update the optimal  solution for each particle $p^i_{best}$ and the global optimal  solution for the particle swarms  $g_{best}$
   \For{$s=1, m$}
  \If { $\Phi(X_s)$ $>$  $\Phi(p^s_{best})$}
     \State $\Phi(p^s_{best})$, $p^s_{best}$  $\leftarrow$ $\Phi(X_s)$, $X_s$
     \EndIf
       \If { $\Phi(X_s)$ $>$  $\Phi(g_{best})$}
     \State  $\Phi(g_{best})$, $g_{best}$  $\leftarrow$ $\Phi(X_s)$, $X_s$
         \EndIf
         \If { $\Phi(X_s)$ $=$  $g_{best}$}
         \State  Re-initialize the particle
            \EndIf
  \EndFor

      \State Update speed  $v$ and position $x$
      \State  Select the top $K$ positions of each particle based on the ranking of their current value and reset them to 1; the other positions are reset to 0.
      \EndFor

    \Return  the $K$ links of the largest  damage
  \end{algorithmic}
\end{algorithm}

\subsection{Greedy attack based on link centrality}
In a simple greedy attack, we can calculate the damage of each link based on Eq. (3) and then select the $K$ links of the largest damage. However, it is time consuming to calculate the damage of all links. Therefore, we use our link centrality measurement to filter the links so that the relatively important links are left for consideration. Specifically, we rank  all of the links based on their link centrality values. Then, we calculate the damage of the  top $L$ ($ML\%\geq K$) percent of links in the ranking,  and select  the $K$ links of the largest damage.

Note that the parameters of our link centrality measure affect the ranking of links and therefore the selection of the $K$-link set. The better parameters correspond to a better ranking, and then  $L$ can be  much smaller.
  %corresponding to the case that the link on top of the ranking has the largest damage. For the sake of reducing computing complexity, we use these parameters as the approximation of the optimal parameters of the greedy attack.
 % {\color{red}{Since the time complexity of the cascade fault is difficult to calculate, the time complexity of the algorithm is the number of iterations of the cascade fault.}}
The  time complexity of the link centrality based greedy attack (LC-GA) is $O(ML\%)$ without considering the duration of single cascading failure.
Note that  we  can also use  PSO to search the optimal parameters in this algorithm, but the computational cost is very large.
The pseudocode of LC-GA is given in Algorithm 2.
  \begin{algorithm}[htb]
    \caption{Link centrality based greedy attack (LC-GA)}
  \label{alg:Framwork}
  \begin{algorithmic}
    \Require
     Adjacency matrix $G$,  power generator node set $Q$, and  integer $K$
    \Ensure
   The total number of selected links is fixed to be $K$.
    \\ \textit{Initialisation}: Give a parameter set
    \State  Select the top $L$\% links based on the ranking of link centrality values
    \For{$j=1, ML\%$}
    \State Calculate the damage  ${\Phi}$ for each selected link
      \EndFor

\Return  the $K$ links of the largest damage
  \end{algorithmic}
\end{algorithm}

\subsection{Link centrality based optimal attack }
The PSO-OA algorithm does not use the topological and electrical properties of power grids. Thus, it is supposed to be not efficient. The LC-GA algorithm does use the topological and electrical properties, but it needs to calculate the link centrality and the damage of many  links, which has large computational cost.  Here, we employ  PSO to search the optimal parameters so that removing the top $K$ links in the ranking of link centrality leads to the maximum damage. The set of the top $K$ links  is thus an approximate  solution  of the ${KLS}$ problem.
 The time complexity of this link centrality based optimal attack (LC-OA) is $O(m*iter_{max})$ without considering the duration of single cascading failure. The  pseudocode of this algorithm is provided  in Algorithm 3.
 \begin{algorithm}[htb]
  \caption{Link centrality based optimal attack (LC-OA)}
  \label{alg:Framwork}
  \begin{algorithmic}
    \Require
     Adjacency matrix $G$, power generator node set $Q$, and  integer $K$
    \Ensure
    The total number of selected links is fixed to be $K$.
         \\ \textit{Initialisation}: Randomly generate $m$ particles; each particle   contains 2 elements which are  randomly set to values in the range  [-1,1].
          \For{i=1, {$iter_{max}$}}
          %\State The number of links  per attack is set to $K$.
          \State Calculate the centrality of each link based on Eq. (4)
          \State Calculate the damage of removing the $K$ links of the largest centrality
          \State Update $p_{best}$ and $g_{best}$  based on the  PSO algorithm
         \EndFor
         %\State Calculate the centrality of each link based on $g_{best}$

\Return the $K$ links of the largest centrality
  \end{algorithmic}
\end{algorithm}

\section{Simulation results}
Based on MATLAB, we do simulation experiments to validate our link centrality measure and compare the performance of  proposed attack algorithms. We use the standard  IEEE bus test data \cite{bworld} including IEEE 118 bus, 145 bus, and 162 bus. We randomly set 10 percent of nodes to be generator nodes.
The related parameters of the experiments are  given in Table 1.

%\begin{table}[htb]
%\caption{}\label{tab1}
%\setlength{\tabcolsep}{1mm}{
%\begin{tabular} {cccc}
%\hline
%\multicolumn{4}{c}{ Network paramete}\\
%\hline
%Network & N & E & Power generation node set  \\
%IEEE118  & 118  & 179 &  [17 18 30 36 46 61 72 76 81 93 98 104 105 112] \\
%IEEE145  & 145  & 442 & [86,74,25,45,95,7,127,108,3,115,91,117,37,90,46]\\
%IEEE162  & 162  & 280 & [87,126,108,142,83,89,141,22,32,149,143,44,116,2,153,154,58]\\
%\hline
%\end{tabular}}
%\end{table}

\begin{table}[H]
\caption{}\label{tab1}
\setlength{\tabcolsep}{7mm}{
\begin{tabular} {cccc}
\hline
\multicolumn{4}{c}{Parameter settings in the experiments}\\
\hline
$m$&10&$c_2$&0.7\\
$iter_{max}$&30&$\alpha$&0.2\\
$w_0$&0.96&$\beta$&0.2\\
$c_1$&0.7&$L$&50\%\\
\hline
\end{tabular}}
\end{table}

\subsection{Single link attack}
\begin{figure}[H]
\centering
\includegraphics[width=2.5in]{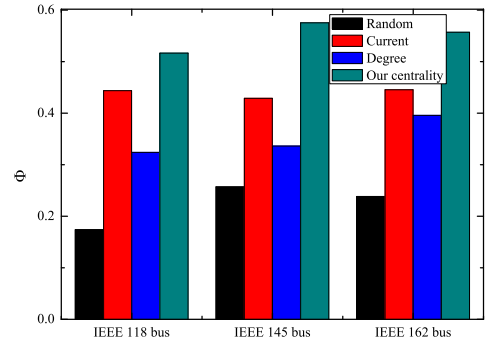}
\caption{
The results of damage $\Phi$ caused by the removal of the most critical link associated  with different centrality measures on three IEEE bus test data.
}
\label{fig1}
\end{figure}
First, we study the single link attack,  $K=1$, which is removing the most critical link in terms of damage. The link degree, link currency, and our link centrality are used respectively to determine the critical link.  For our link centrality measure, we use  PSO to find the optimal solution (see LC-OA). The results are shown in Fig. 1, which are the average of 100 independent runs.   We can see that for all the three IEEE bus test data, the random removal has the smallest damage, and the damage of our link centrality is larger than the link degree and link current. Moreover,  link current is more efficient than  link degree, which indicates that we should focus  on electrical features more than  topological features in the vulnerability assessment of power grids.
 %The optimal parameters for the  IEEE 118 bus, 145 bus, and 162 bus data are $[0.4376,0.0882], [0.6316,0.1430]$, and $[0.2318,0.9568]$, respectively, where the parameter of link current is always larger than link degree.  This further indicates that the electrical features play a more important role than the topological features in find the critical links to cascading failures.

%\begin{figure}
%\includegraphics[width=\textwidth]{threeedges.png}
%\caption{Failure Unilateral Method Based on Random Selection, Single Attribute and Comprehensive Indicators.} \label{fig1}
%\end{figure}

\subsection{Multiple link attack }

%\begin{figure*}[H]
%\centering
%\subfloat[Case I]{\includegraphics[width=1in]{fig100.eps}%
%\hfil
%\subfloat[Case II]{\includegraphics[width=1in]{fig200.eps}%
%\label{fig_second_case}}
%\caption{Simulation results for the network.}
%\label{fig_sim}
%\end{figure*}
\begin{figure}[htb]
\centering
\includegraphics[width=3.5in]{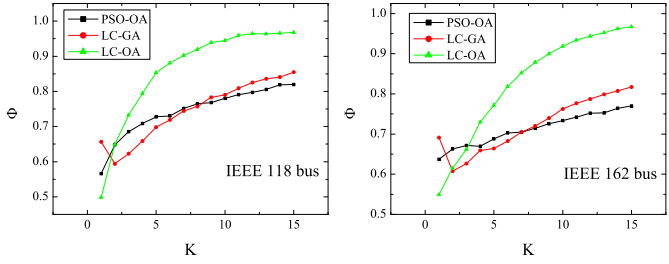}
\caption{
The results of damage $\Phi$ vs. number of removed links $K$ for the three attack algorithms on IEEE 118 bus and 162 bus data.
}
\label{fig2}
\end{figure}
\begin{figure}[htb]
\centering
\includegraphics[width=3.5in]{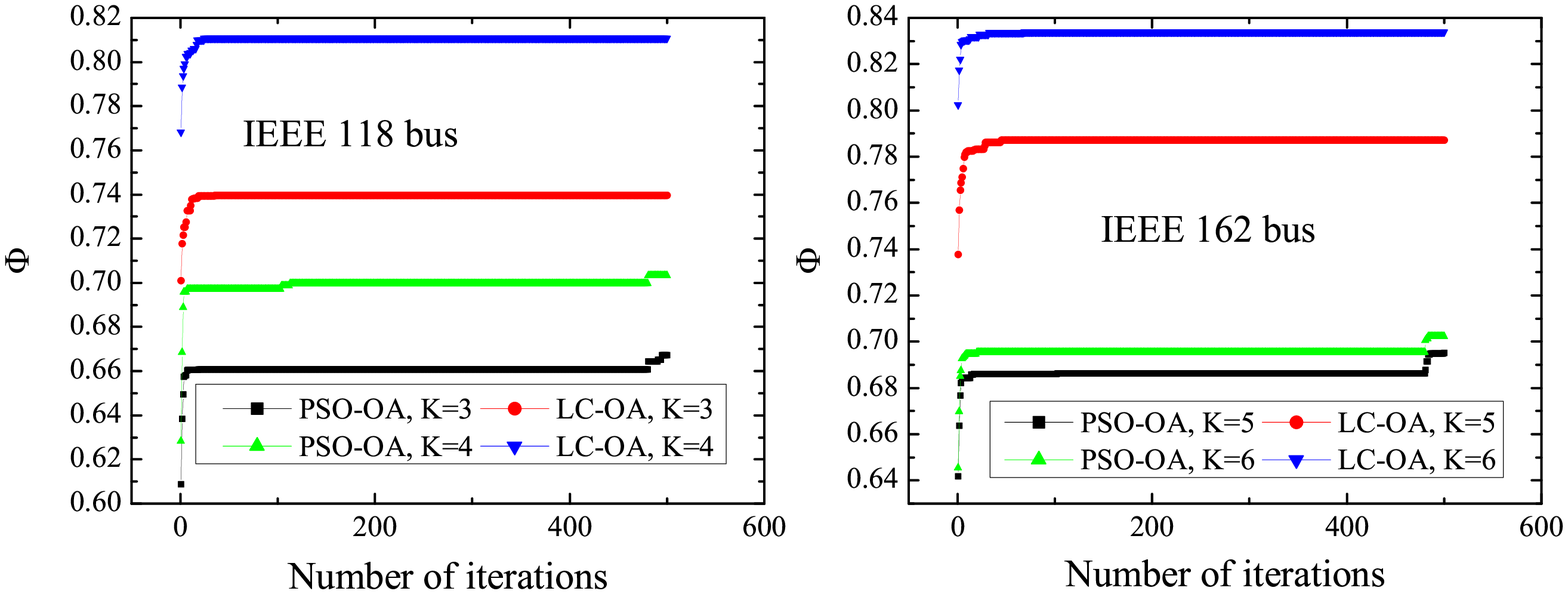}
\caption{The results of damage $\Phi$ vs. number of iterations for LC-OA and PSO-OA on IEEE 118 bus and 162 bus data.  }
\label{fig3}
\end{figure}
Furthermore, we study the multiple link attack,  $k>1$,  to compare the performance of PSO-OA, LC-OA, and LC-GA. For LC-GA,  parameters $h_1$ and $h_2$ are both set to  1.   The results of  damage are provided in Figs. 2 and  3.     In Fig. 2, for all the IEEE bus  data, the damage generally increases with the number of removed links for all the algorithms. This is  because the more links removed at the beginning result in  a wider range of cascading failure. An exception is that for LC-GA, the damage decreases when the number of removed links increases from 1  to 2.  The reason is that LC-GA is essentially  a greedy algorithm so that its solution of $K=2$ is not the optimal one, while for  $K=1$, the solution could be the optimal one.    Moreover, we see that the damage of LC-OA is much larger than  PSO-OA and LC-GA.
In real situations, the numbers of particles and iterations  are limited, as the parameter settings in our experiments. In this case,  LC-OA can find a better solution than PSO-OA for the multiple link attack problem. The damage of  LC-GA is smaller than LC-OA, since the former is a greedy algorithm, while the latter is based on global optimization essentially. Note that LC-GA is sometimes better than PSO-OA as shown in the results of $K>6$ on IEEE 162 bus data. This further validates our link centrality measure.
In Fig. 3, we see that for a given $K$, the damage of LC-OA increases and converges faster than PSO-OA with the growth of number of iterations. This further demonstrates  that LC-OA is more efficient than PSO-GA. Moreover,  we see that  PSO-OA is prone to fall into the local optimum and needs many iterations to jump out of it.

\section{Conclusion}
In summary, we study the vulnerability assessment of power grids under cascading failures. We define the $K$-link attack problem and prove that  it is NP-complete. Particularly, we propose a link centrality measure by combining the link degree with link current. With this centrality, we develop two attack algorithms, which are the link centrality based greedy attack (LC-GA) and the link centrality based optimal attack (LC-OA). We evaluate our link centrality measure and its related attack algorithms on standard IEEE bus test data. Simulation results show that our link centrality measure performs better than the link degree and link current in  identifying the optimal link in the single link attack scenario. Furthermore, in the multiple link attack problem,  LC-OA is much more efficient than  LC-GA and the traditional PSO based optimal attack (PSO-OA) algorithm. LC-GA has  lower computational complexity than LC-OA and PSO-OA with a decent attack efficiency.

% or


%\appendices

%\section{}

%\section*{Acknowledgment}


\ifCLASSOPTIONcaptionsoff
  \newpage
\fi

\end{document}